\documentclass[a4paper,onecolumn,accepted=2023-01-09]{quantumarticle}
\pdfoutput=1 

\usepackage{hyperref}
\usepackage{amsmath, amsthm, amsfonts}
\usepackage[numbers]{natbib}

\newtheorem{definition}{Definition}
\newtheorem{theorem}{Theorem}
\newtheorem{lemma}{Lemma}
\newtheorem{corollary}{Corollary}

\newtheorem{proposition}{Proposition}

\newcommand{\bra}[1]{\ensuremath{\langle #1 |}}
\newcommand{\ket}[1]{\ensuremath{| #1 \rangle}}
\newcommand{\braket}[2]{\ensuremath{\langle #1 | #2 \rangle}}
\newcommand{\ketbra}[2]{\ensuremath{| #1 \rangle \langle #2 |}}
\newcommand{\F}{\mathbb{F}}
\newcommand{\X}{\sigma_1}
\newcommand{\Z}{\sigma_3}
\newcommand{\Y}{\sigma_2}

\newcommand{\I}{\mathbf{1}}

\newcommand{\ind}[1]{\I \left( #1 \right)}
\newcommand{\C}{\mathbb{C}}

\newcommand{\functiong}{g_{B_n(0,r)}}
\newcommand{\functionf}{f_{B_n(0,r)}}

\newcommand{\colvec}[2]{\ensuremath{\begin{pmatrix}
			#1 \\ #2
\end{pmatrix}}}
\newcommand{\colvecthree}[3]{\ensuremath{\begin{pmatrix}
			#1 \\ #2 \\ #3
\end{pmatrix}}}
\newcommand{\rowvec}[2]{\ensuremath{\begin{pmatrix}
			#1 & #2
\end{pmatrix}}}

\begin{document}
\title{QKD parameter estimation by two-universal hashing}
\author{Dimiter Ostrev}
\orcid{0000-0002-4098-0969}
\affiliation{Institute of Communications and Navigation, German Aerospace Center, Oberpfaffenhofen, 82234 Weßling, Germany
}
\thanks{This work was performed while the author was at the Interdisciplinary Centre for Security
Reliability and Trust, University of Luxembourg,
Esch-sur-Alzette, L-4364, Luxembourg. This work is supported by the Luxembourg National Research Fund, under CORE project Q-CoDe (CORE17/IS/11689058). }
\date{}
\maketitle

\begin{abstract}
	This paper proposes and proves security of a QKD protocol which uses two-universal hashing instead of random sampling to estimate the number of bit flip and phase flip errors. This protocol dramatically outperforms previous QKD protocols for small block sizes. More generally, for the two-universal hashing QKD protocol, the difference between asymptotic and finite key rate decreases with the number $n$ of qubits as $cn^{-1}$, where $c$ depends on the security parameter. For comparison, the same difference decreases no faster than $c'n^{-1/3}$ for an optimized protocol that uses random sampling and has the same asymptotic rate, where $c'$ depends on the security parameter and the error rate. 
\end{abstract}

\section{Introduction}

\subsection{Motivation}

Quantum Key Distribution allows two users, Alice and Bob, to agree on a shared secret key using an authenticated classical channel and a completely insecure quantum channel. There are information theoretic security proofs for QKD protocols (for example \cite{shor2000simple,renner2005security,koashi2009simple,fung2010practical,bouman2010sampling,tomamichel2012tight,leverrier2017largely} among many others). Quantum key distribution has also been realized experimentally and is commercially available. The rare combination of information theoretic security and practical achievability has attracted considerable attention to QKD. 

A QKD protocol has several important parameters:
\begin{enumerate}
\item Block size: the number of pairs of qubits that Alice and Bob receive. Following \cite[Part 1]{leverrier2017largely}, this paper considers entanglement based protocols and defines the block size as the number of qubits after sifting. 
\item Output size: the number of bits of secret key that the protocol produces. 
\item Key rate: the ratio of output size to block size. The higher the key rate is, the more efficiently the protocol converts the available quantum resource to a secret key. 
\item Security level: the distance of the output from an ideal secret key. The lower the security level, the better the guarantee that no future evolution of the protocol output and adversary registers will be able to distinguish between the output and an ideal key. 
\item Robustness: the amount and type of noise that the protocol can tolerate without aborting. In particular, the QKD protocol should be able to tolerate at the very least the imperfections of whatever quantum channel and entanglement source are used to implement the protocol.  
\end{enumerate}

Existing QKD protocols and security proofs exhibit trade-offs between these parameters: improving the security or robustness of the protocol worsens the key rate. These trade-offs are particularly severe when the block size is small. The phenomenon that the key rate of a QKD protocol deteriorates significantly for small block sizes has been called finite size effect \cite[Sections II-C and IX]{pirandola2020advances}.

The finite size effect has practical consequences in cases when the distribution of entangled quantum states is particularly difficult. As an example, consider the problem of using QKD between users who are far apart on the surface of the earth. The Micius satellite experiment \cite{yin2020entanglement} tried to solve this problem by using a satellite to distribute entangled photon pairs to two ground stations that are 1120km apart. However, sending entangled photon pairs from space to earth is difficult. In the Micius experiment, several nights of good weather had to pass until the ground stations accumulated sifted block size 3100. The error rate that the ground stations needed to tolerate was $4.51\%$. Reference \cite{lim2021security} performed a state-of-the-art security analysis on this data, and concluded that security levels better than around $10^{-6}$ lead to no secret key at all, while at security level $10^{-6}$, only six bits of secret key are extracted. 

The output size and security level achieved in this example are not sufficient for cryptographic applications. This provides the motivation for the present paper. Are there QKD protocols and security proofs that achieve a combination of small block size with output size and security level sufficient for cryptographic applications?

\subsection{Contributions} 

This paper presents the two-universal hashing QKD protocol and proves its security. The two-universal hashing QKD protocol is an entanglement based protocol with block size $n$, that can tolerate any combination of up to $r$ bit flip errors and up to $r$ phase flip errors, and at the end extract $n - 2 \lceil n h(r/n) + 2\log_2(1/\epsilon) + 5 \rceil $ secret key bits, that are $\epsilon$ close to an ideal secret key. 

For small block sizes, the two-universal hashing QKD protocol dramatically outperforms protocols of the BBM92 type. To illustrate, consider again the security analysis developed in the sequence of papers \cite{tomamichel2012tight,leverrier2017largely,lim2021security} applied to the Micius satellite example. 
\begin{enumerate}
\item Fix the tolerated error rate at $4.51\%$, the security level at $10^{-6}$ and the output size at 6 bits. The BBM92 type protocol with the security proof developed in \cite{tomamichel2012tight,leverrier2017largely,lim2021security} requires block size 3100. The two-universal hashing protocol requires block size 200. 
\item Fix the block size at 3100 and fix the error rate at $4.51\%$. The BBM92 type protocol with the security proof developed in \cite{tomamichel2012tight,leverrier2017largely,lim2021security} can extract $6$ secret key bits with security level $10^{-6}$. The two universal hashing protocol can extract $385$ secret key bits with security level $10^{-80}$. 
\end{enumerate}

The advantage of the two-universal hashing QKD protocol is particularly noticeable for small block sizes; however, it is not limited to them. For fixed error rate $\delta=r/n$ and fixed security parameter $\epsilon$, the asymptotic rate of this protocol is $1-2h(\delta)$, and the deviation of finite from asymptotic rate is between $(4 \log_2(1/\epsilon)+10)/n$ and $(4 \log_2(1/\epsilon)+12)/n$. By contrast, the deviation of finite from asymptotic key rate for the BBM92 type protocol with the security proof \cite{tomamichel2012tight,leverrier2017largely,lim2021security} is of the form $c n^{-1/3}$, where $c$ depends on the tolerated error rate and the security level. A discussion of the reasons for the difference in parameter trade-offs follows. 

In the BBM92 type protocol, a random subset of $n_{pe}$ positions is measured and the outcomes are publicly compared. If the error rate on this subset is below some threshold $\delta$, then parameter estimation accepts and outputs the promise that the error rate on the remaining positions is at most $\delta+\nu$, where $\nu$ is the gap between observed and inferred error rate. The failure probability for parameter estimation scales roughly as $exp(-4n_{pe}v^2)$. To get a sense of this scaling, suppose that the target failure probability is $e^{-100}$ and that the target gap is $\nu=0.01$. Then, $n_{pe}$ has to be chosen to be $250000$, clearly orders of magnitude more than can be afforded for block sizes around 1000 or 10000. Now, consider the rest of the protocol. Information reconciliation and privacy amplification have to operate with the promise that the error rate is at most $\delta+\nu$. Thus, to extract a secret key, information reconciliation and privacy amplification have to sacrifice a substantial number of positions beyond the initial $n_{pe}$ sacrificed for parameter estimation. 

By contrast, in the two-universal hashing protocol, $2k$ ebits are sacrificed for parameter estimation. If the test passes, then Alice and Bob know that the post-parameter-estimation state is a particular Bell state of $n-2k$ ebits; thus, Alice and Bob do not need to sacrifice any further ebits for information reconciliation and privacy amplification. Moreover, the scaling of the failure probability for parameter estimation with the number of sacrificed ebits does not have the $\nu^2$ coefficient in front of the number of sacrificed ebits. 

\subsection{Related work}

The present paper builds on a number of previous ideas. 

In classical information theory, random linear functions have been used to obtain ensembles of error correcting codes since the 1950s: see for example \cite[Section 2.1]{gallager1963low} where the idea was attributed to Elias \cite{elias1955coding2}. Random linear functions are also a special case of two-universal hash functions and can be used to authenticate classical messages \cite{carter1979universal}. Further, \cite[Theorem 6]{brassard1993secret}, \cite[Section 6.3.2]{renner2005security} observed that two-universal hash functions can be used to achieve information reconciliation with minimum leakage.

In quantum information theory, \cite{bennett1996mixed} used a variant of linear two-universal hash functions to perform entanglement purification. \cite{lo1999unconditional} applied the technique of \cite{bennett1996mixed} to construct an LOCC protocol by which Alice and Bob can verify that a state they received from the adversary was in fact $n$ perfect EPR pairs. \cite{shor2000simple} observed that when parameter estimation has already been performed by random sampling, arguments related to quantum CSS codes \cite{calderbank1996good,steane1996multiple} can be used to prove security of a QKD protocol. \cite{koashi2009simple} presented an interesting generalization of the proof technique of \cite{shor2000simple} that works also in the case of imperfect devices. \cite{bouman2010sampling} translated the guarantees of classical random sampling to the quantum case and used this to obtain a QKD security proof. \cite{tomamichel2012tight,fung2010practical} focused on the performance of QKD for small block sizes, and optimized their protocols by using random sampling to estimate the number of errors in only one of the measurement bases, while using a two-universal hash in the information reconciliation phase to ensure correctness. \cite{leverrier2017largely} developed the proof idea of \cite{tomamichel2012tight} with much greater mathematical rigour. \cite{lim2021security} proved a better random sampling tail bound and thus obtained better parameter trade-offs than \cite{leverrier2017largely}. 

From the references above, closest to the present paper is \cite{lo1999unconditional}. The current paper develops further the ideas in \cite{lo1999unconditional} in the following ways: 
\begin{enumerate}
\item Some mathematical details in the proof of \cite{lo1999unconditional} were skipped, other details were entrusted to the papers on stabilizer codes and entanglement purification. Further, \cite{lo1999unconditional} did not discuss composable security\footnote{Frameworks for composable security were not yet invented at the time of \cite{lo1999unconditional}.} and did not give any explicit bounds on the achievable parameters for specific finite block sizes. The present paper gives a detailed, rigorous and self-contained proof of composable security, and gives explicit formulas for the achievable parameters at any finite block size. 
\item \cite{lo1999unconditional} proposed a QKD protocol that employed full error correction with a stabilizer error correcting code, followed by their verification subroutine; thus the quantum phase of their protocol required the ability to implement stabilizer error correcting codes. The present paper relies on quantum CSS codes to simplify the quantum phase of the protocol as much as possible.  
\end{enumerate}

\subsection{Structure of the paper}

The rest of this paper is structured as follows: Section \ref{sec:Preliminaries} introduces material that is needed to present and prove the security of the two universal hashing QKD protocol, including the security and robustness criteria for QKD protocols, a number of useful lemmas related to the stabilizer formalism, and a number of useful lemmas about two-universal hashing and random matrices over the field with two elements. Section \ref{sec:QKDProtocolAndItsSecurity} presents the two-universal hashing QKD protocol and shows that it is secure and robust. Section \ref{sec:ComparisonWithPreviousWork} shows that for fixed security level and tolerated error rate, the finite key rate converges to the asymptotic rate as $cn^{-1}$ for two-universal hashing and as $c'n^{-1/3}$ for the BBM92 protocol with the security proof developed in \cite{tomamichel2012tight,leverrier2017largely,lim2021security}. Section \ref{sec:Conclusion} concludes and gives some open problems. 

\section{Preliminaries}\label{sec:Preliminaries}

This section presents definitions and results that are used to state and prove the main result on the security and robustness of the two-universal hashing protocol. Subsection \ref{sec:SecurityCriterion} recalls the standard security criterion for QKD. Then, subsection \ref{subsec:ThePauliGroupAndTheBellBasis} contains a number of lemmas related to the stabilizer formalism; these are used during the security proof. Finally, subsection \ref{sec:ComputingCertainFunctions} contains lemmas related to two-universal hashing. Subsection \ref{sec:ComputingCertainFunctions} also discusses an application of two-universal hashing to approximately compute certain functions from partial information about the input; this is used during the security proof. 

\subsection{Security and robustness of quantum key distribution}\label{sec:SecurityCriterion}

This section recalls the security and robustness criteria from \cite{renner2005security} that ensure that the key produced by QKD can be used in any application. See \cite{portmann2014cryptographic} for a proof of the equivalence of this security criterion and security in the Abstract Cryptography framework for composable security. 

As is common in the QKD literature, this paper assumes that the adversary Eve is active in the quantum phase of the protocol but remains passive during the classical phase, i.e. Eve eavesdrops the classical communication but does not attempt to modify or block it. Under this assumption, an entanglement-based QKD protocol is a completely positive trace preserving map that transforms input states $\rho_{ABE}$ of Alice, Bob and Eve into output states $\tilde{\rho}_{W_AW_BCE}$, where $W_A,W_B$ are registers containing Alice and Bob's output: a secret key or indication $\bot$ of protocol abort, and where $C$ is a register containing a transcript of the classical communication between Alice and Bob. 

Since registers $W_A,W_B$ contain classical values, the final state $\tilde{\rho}_{W_AW_BCE}$ can be decomposed as 
\begin{equation}
\tilde{\rho}_{W_AW_BCE} = \ketbra{\bot \bot}{\bot \bot}_{W_AW_B} \otimes \tilde{\rho}_{CE}(\bot) + \sum_{w_A, w_B} \ketbra{w_Aw_B}{w_Aw_B}_{W_AW_B} \otimes \tilde{\rho}_{CE}(w_A,w_B)
\end{equation}
This decomposition is used to formulate the definition of security:
\begin{definition}\label{def:EpsilonSecurity}
A QKD protocol is $\epsilon$ secure if for all input states $\rho_{ABE}$, the output state $\tilde{\rho}_{W_AW_BCE}$ is $\epsilon$-close in trace distance to the corresponding ideal state \begin{equation} \ketbra{\bot \bot}{\bot \bot}_{W_AW_B} \otimes \tilde{\rho}_{CE}(\bot) + \sum_w \frac{1}{|W|}\ketbra{ww}{ww}_{W_AW_B} \otimes (\tilde{\rho}_{CE} - \tilde{\rho}_{CE}(\bot)) \end{equation} where $|W|$ denotes the size of the secret key space. 
\end{definition}

Alternatively, $\epsilon$-security can be further subdivided into requirements for secrecy and correctness: 
\begin{definition}\label{def:EpsilonCorrectness}
A QKD protocol is $\epsilon$ correct if for all input states $\rho_{ABE}$, the probability \begin{equation} Pr(W_A \neq W_B) = \sum_{w_A \neq w_B} Tr(\tilde{\rho}_{CE}(w_A,w_B)) \end{equation} that Alice and Bob accept and output different keys is bounded by $\epsilon$. 
\end{definition}

\begin{definition}
Alice's key is $\epsilon$ secret if for all input states $\rho_{ABE}$, the reduced output state $\tilde{\rho}_{W_ACE}$ is $\epsilon$-close in trace distance to the corresponding ideal state \begin{equation} \ketbra{\bot}{\bot}_{W_A} \otimes \tilde{\rho}_{CE}(\bot) + \sum_w \frac{1}{|W|} \ketbra{w}{w}_{W_A} \otimes (\tilde{\rho}_{CE} - \tilde{\rho}_{CE}(\bot)) \end{equation}
\end{definition}

The following lemma establishes the relation between security and correctness plus secrecy:
\begin{lemma}\label{lemma:RelationBetweenSecurityAndCorrectnessPlusSecrecy}
If a QKD protocol is $\epsilon$ secure, then it is $\epsilon$ correct and Alice's key is $\epsilon$ secret. Conversely, if the protocol is $\epsilon$ correct and Alice's key is $\delta$ secret, then the protocol is $\epsilon + \delta$ secure. 
\end{lemma}

\begin{proof}
The forward direction follows from monotonicity of the trace distance and its interpretation as distinguishing advantage. The reverse direction follows by considering the hybrid state \begin{equation} \ketbra{\bot \bot}{\bot \bot}_{W_AW_B} \otimes \tilde{\rho}_{CE}(\bot) + \sum_{w} \ketbra{ww}{ww}_{W_AW_B} \otimes \sum_{w'}\tilde{\rho}_{CE}(w,w')  \end{equation} and the triangle inequality.
\end{proof}

Next, note that in the standard definition of QKD security (Definition \ref{def:EpsilonSecurity}) the ideal state, beyond being $\epsilon$ close to the real state, satisfies the following additional conditions:
\begin{enumerate}
\item The probabilities of accepting and rejecting are the same for the real and ideal state. 
\item The real and ideal state differ only in the accept case. 
\item The sub-normalized reduced density matrix of registers $C,E$ in the accept case is equal to $\tilde{\rho}_{CE} - \tilde{\rho}_{CE}(\bot)$ for both the real and the ideal state. 
\end{enumerate}
Now, suppose an ideal state is found that is $\epsilon$ close to the real state, but which does not necessarily satisfy these additional conditions. This suffices to demonstrate security:
\begin{lemma}\label{lemma:RelaxingImplicitAssumptionsInEpsilonSecurity}
Suppose that for all input states $\rho_{ABE}$, there exist positive $\sigma_{CE}^{accept}$ and $\sigma_{CE}^{reject}$ such that $Tr(\sigma_{CE}^{accept}) + Tr(\sigma_{CE}^{reject})=1$ and such that the output state $\tilde{\rho}_{W_AW_BCE}$ is $\epsilon$-close in trace distance to \begin{equation} \ketbra{\bot\bot}{\bot\bot}_{W_AW_B}\otimes\sigma_{CE}^{reject} + \sum_{w}\frac{1}{|W|} \ketbra{ww}{ww}_{W_AW_B} \otimes \sigma_{CE}^{accept} \end{equation} Then, the protocol is $2\epsilon$ secure. 
\end{lemma} 

\begin{proof}
By assumption, 
\begin{equation} 
\frac{1}{2}\|\tilde{\rho}_{CE}(\bot) - \sigma_{CE}^{reject}\|_1 \\ + \frac{1}{2} \sum_{w_A,w_B} \left\| \tilde{\rho}_{CE}(w_A,w_B) - \frac{1}{|W|}\mathbf{1}(w_A=w_B)\sigma_{CE}^{accept} \right\|_1 \leq \epsilon 
\end{equation} 
From the triangle inequality it follows that 
\begin{equation} 
\frac{1}{2}\left(\|\tilde{\rho}_{CE}(\bot) - \sigma_{CE}^{reject}\|_1 +  \left\| \tilde{\rho}_{CE}-\tilde{\rho}_{CE}(\bot) -\sigma_{CE}^{accept} \right\|_1\right) \leq \epsilon 
\end{equation} 
The lemma then follows by another application of the triangle inequality. 
\end{proof}

Finally, note that a protocol that always aborts is secure, but not useful. For a useful QKD protocol, the probability of acceptance is bounded below by $1-\delta$ for some $\delta \in (0,1)$ on a suitable class of input states. In the present paper, robustness of the two-universal hashing protocol is shown by giving explicit bounds on the probability of acceptance as a function of the input state. 

\subsection{The Pauli group and the Bell basis}\label{subsec:ThePauliGroupAndTheBellBasis}

This section presents a number of useful lemmas related to the stabilizer formalism \cite{gottesman1996class,calderbank1997quantum}.

Denote the Pauli matrices by \begin{equation} \X = \begin{pmatrix}
0 & 1 \\ 1 & 0
\end{pmatrix}, \;\;\; \Y = \begin{pmatrix}
0 & -i \\ i & 0
\end{pmatrix}, \;\;\;  \Z = \begin{pmatrix}
1 & 0 \\ 0 & -1
\end{pmatrix} \end{equation}

Let $\F_2$ denote the field with two elements, $\F_2^n$ denote the $n$-dimensional vector space over this field, and $\F_2^{n \times m}$ denote the space of $n$ by $m$ matrices over $\F_2$. For a row vector $u \in \F_2^{1 \times n}$, denote \begin{equation} \X^u = \X^{u_1} \otimes \dots \X^{u_n}, \;\;\; \Z^u = \Z^{u_1} \otimes \dots \otimes \Z^{u_n} \end{equation}

The Pauli group on $n$ qubits is \begin{equation}  G_n = \{\omega \X^u \Z^v: \omega \in \{\pm1, \pm i\}, u,v \in \F_2^{1 \times n}\} \end{equation} Matrix multiplication of elements of $G_n$ can be performed in terms of $u,v,\omega$: \begin{equation} (\omega \X^u \Z^v)(\omega' \X^{u'} \Z^{v'}) = \omega\omega'(-1)^{v \cdot u'} \X^{u+u'} \Z^{v+v'} \end{equation} This also shows that the map $\mathcal{F} : G_n \rightarrow \F_2^{1 \times 2n}$ given by \begin{equation} \mathcal{F}(\omega \X^u \Z^v) = \rowvec{u}{v} \end{equation} is a group homomorphism. 

Any element of the Pauli group squares to either $I$ or $-I$; any two elements $g,g'$ of the Pauli group satisfy \begin{equation} g g' = (-1)^{\mathcal{F}(g) \mathcal{S} \mathcal{F}(g')^T} g'g\end{equation}  where $\mathcal{S} \in \F_2^{2n \times 2n}$ is the matrix with block form \begin{equation} \mathcal{S} = \begin{pmatrix}
0 & I_n \\ I_n & 0
\end{pmatrix} \end{equation}

Say that a tuple of elements of the Pauli group \begin{equation} \vec{g}= \colvecthree{g_1}{\vdots}{g_m} \end{equation} is independent if the row vectors $\mathcal{F}(g_i) \in \F_2^{1 \times 2n}$  are linearly independent. Given such an independent tuple and given any $x \in \F_2^m$, it is possible to find $g \in G_n$ such that \begin{equation} \forall i, \;\; g g_i = (-1)^{x_i} g_i g \end{equation} by solving the corresponding linear system of equations over $\F_2$.

A tuple  of independent commuting self-adjoint elements of the Pauli group $\vec{g} = (g_1, \dots g_m)^T$ defines a projective measurement on its joint eigenspaces. The measurement outcomes can be indexed by $x \in \F_2^m$ and the corresponding projections are given by \begin{equation} P(\vec{g}, x) = 2^{-m} \prod_{j=1}^m (I + (-1)^{x_j} g_j) \end{equation} The projections $P(\vec{g},x)$ form a complete set of orthogonal projections. The elements of the Pauli group map these projections to each other under conjugation, as can be seen from Lemma \ref{lemma:ActionOfPauliGroupOnMeasurementProjections} below. Therefore, the projections $P(\vec{g},x)$ all have the same rank $2^{n-m}$. 

\begin{lemma}\label{lemma:ActionOfPauliGroupOnMeasurementProjections}
	For all tuples $\vec{g}= (g_1 \dots g_m)^T $ of independent commuting self-adjoint elements of $G_n$, for all $h \in G_n$, for all $x \in \F_2^m$, \begin{equation} P(\vec{g}, x) h = h P(\vec{g}, x + \mathcal{F}(\vec{g}) \mathcal{S} \mathcal{F}(h)^T) \end{equation} where \begin{equation} \mathcal{F}(\vec{g}) = \colvecthree{\mathcal{F}(g_1)}{\vdots}{\mathcal{F}(g_m)}\end{equation} is the matrix with rows $\mathcal{F}(g_1), \dots, \mathcal{F}(g_m)$. 
\end{lemma}

\begin{proof}
	\begin{multline}
	P(\vec{g},x)h = 2^{-m} \left(\prod_{j=1}^m (I + (-1)^{x_j} g_j) \right) h \\
	= 2^{-m} h \left( \prod_{j=1}^m (I + (-1)^{x_j + \mathcal{F}(g_j) \mathcal{S} \mathcal{F}(h)^T}g_j)\right) = hP(\vec{g}, x + \mathcal{F}(\vec{g}) \mathcal{S} \mathcal{F}(h)^T)
	\end{multline}
\end{proof}

Now, take a tuple $\vec{g}$ of $m$ independent commuting self-adjoint elements, take $k \leq m$ and take a full rank matrix $L \in \F_2^{k \times m}$. The matrix $L$ transforms the tuple $\vec{g}$ to the $k$-tuple \begin{equation} L \vec{g} = L \colvecthree{g_1}{\vdots}{g_m} = \colvecthree{\prod_{j=1}^m g_j^{L_{1j}}}{\vdots}{\prod_{j=1}^m g_j^{L_{kj}}}  \end{equation} The tuple $L\vec{g}$ also consists of independent commuting self-adjoint elements. The transformation of $\vec{g}$ to $L \vec{g}$ satisfies \begin{equation} M (L\vec{g} )  = (ML)\vec{g} \end{equation} for any $\vec{g}$, $L$, $M$ of compatible size. The matrix $\mathcal{F}(L\vec{g})$ can be expressed in terms of the matrix $\mathcal{F}(\vec{g})$:
\begin{equation}
\mathcal{F}(L\vec{g}) = \colvecthree{\mathcal{F}(\prod_{j=1}^m g_j^{L_{1j}})}{\vdots}{\mathcal{F}(\prod_{j=1}^m g_j^{L_{kj}})} 
= \colvecthree{\sum_{j=1}^m L_{1j}\mathcal{F}(g_j)}{\vdots}{\sum_{j=1}^m L_{kj} \mathcal{F}(g_j)} = L \mathcal{F}(\vec{g})
\end{equation}

The measurement projections of $L\vec{g}$ can be expressed in terms of the measurement projections of $\vec{g}$.

\begin{lemma}\label{lemma:RelationOfMeasurementProjectionsForTwoLinearlyRelatedTuples}
	For all $n \geq m \geq k \geq 1$, for all tuples $\vec{g}$ of $m$ independent commuting self-adjoint elements of $G_n$,  for all full rank $L \in \F_2^{k \times m}$, for all $y \in \F_2^k$, \begin{equation} P(L\vec{g},y) = \sum_{x \in \F_2^m:Lx=y} P(\vec{g},x) \end{equation}
\end{lemma}

\begin{proof}
	Take any $i \in \{1,\dots k\}$, any $x \in\F_2^m$ such that $Lx=y$. Then,
	\begin{multline}
	\left(\prod_{j=1}^mg_j^{L_{ij}}\right) P(\vec{g},x) = \left(\prod_{j=1}^m g_j^{L_{ij}}\right)
	 \left(2^{-m}\prod_{j=1}^m (I+(-1)^{x_j}g_j)\right) \\
	 = (-1)^{\sum_{j=1}^m L_{ij}x_j} P(\vec{g},x) = (-1)^{y_i} P(\vec{g},x)
	\end{multline} 
	Then, for any $x\in\F_2^m$ such that $Lx=y$, $P(L\vec{g},y) P(\vec{g},x) = P(\vec{g},x)$ holds. Since $\{P(\vec{g},x):Lx=y\}$ is a collection of $2^{m-k}$ orthogonal projections of rank $2^{n-m}$ and since $P(L\vec{g},y)$ has rank $2^{n-k}$, the lemma follows. 
\end{proof}

The maximally entangled state in $\C^{2^n} \otimes \C^{2^n}$ is \begin{equation} \ket{\psi} = 2^{-n/2} \sum_{z \in \F_2^n} \ket{zz} \end{equation}
The collection \begin{equation} \ket{\psi_{\alpha\beta}} =  I \otimes \X^{\alpha^T} \Z^{\beta^T} \ket{\psi} , \;\; \alpha,\beta \in \F_2^{n}  \end{equation}
is the Bell basis of $\C^{2^n} \otimes \C^{2^n}$. 

First, the maximally entangled state has the properties: 

\begin{lemma}\label{lemma:PropertiesOfMaximallyEntangledState}
	For all matrices $M \in \C^{2^n \times 2^n}$, $M \otimes I \ket{\psi} = I \otimes M^T \ket{\psi}$ and $\bra{\psi} I \otimes M \ket{\psi} = 2^{-n} Tr (M)$. 
\end{lemma}

\begin{proof}
	Follows by expanding $M$ in the computational basis. 
\end{proof}

Pauli group measurements acting on Bell basis states satisfy the following:

\begin{lemma}\label{lemma:ActionOfPauliMeasurementOnBellBasis}
	For all tuples $\vec{g}$ of independent self-adjoint commuting elements of $G_n$ such that the associated projections $P(\vec{g},x)$ have only real entries when expressed as matrices in the computational basis, for all $\alpha,\beta \in \F_2^{n}$, for all $x,y \in \F_2^m$, 
	\begin{equation} (P(\vec{g},x) \otimes P(\vec{g},y))\ket{\psi_{\alpha\beta}} = \ind{x= y+\mathcal{F}(\vec{g}) \mathcal{S}\colvec{\alpha}{\beta}} P(\vec{g},x) \otimes I \ket{\psi_{\alpha\beta}} \end{equation} where for an expression that takes the values true or false, $\I(expression)$ takes the corresponding values 1 or 0. 
\end{lemma}

\begin{proof}
	Follows from Lemma \ref{lemma:ActionOfPauliGroupOnMeasurementProjections} and the relation $M \otimes I \ket{\psi} = I \otimes M^T \ket{\psi}$
\end{proof}

The QKD security proof also uses the following lemma. It gives two equivalent expressions for the projection on the subspace of $\C^{2^n} \otimes \C^{2^n}$ that corresponds to a specific pattern of bit flip errors or a specific pattern of phase flip errors. 

\begin{lemma}\label{lemma:SumOfBellStateProjections}
	For all $n$, for all $\alpha, \beta \in \F_2^{n}$,
	\begin{align}
	\sum_{\beta' \in \F_2^{n}} \ketbra{\psi_{\alpha\beta'}}{\psi_{\alpha\beta'}} &= \sum_{z_A \in \F_2^n} \ketbra{z_A, z_A + \alpha}{z_A, z_A + \alpha} \\
	\sum_{\alpha' \in \F_2^{n}} \ketbra{\psi_{\alpha'\beta}}{\psi_{\alpha'\beta}} &= \sum_{x_A \in \F_2^n} H^{\otimes 2n}\ketbra{x_A, x_A + \beta}{x_A, x_A + \beta}H^{\otimes 2n}
	\end{align}
\end{lemma}

\begin{proof}
	Let $e_1, \dots e_n$ denote the standard basis of $\F_2^{1 \times n}$. For $i \in \{1,3\}$ and $R \in \{A,B\}$, let $\vec{\sigma}^R_i$ denote the tuple $\sigma_i^{e_1}, \dots \sigma_i^{e_n}$ acting on register $R$, and let $\vec{\sigma}_i^{AB}$ denote the tuple $\sigma_i^{e_1}\otimes \sigma_i^{e_1}, \dots, \sigma_i^{e_n} \otimes \sigma_i^{e_n}$. Note that for all $\alpha, \beta$, 
	\begin{equation}
	\ketbra{\alpha\beta}{\alpha\beta}_{AB} = P\left(\colvec{\vec{\sigma}_3^A}{\vec{\sigma}_3^B}, \colvec{\alpha}{\beta}\right) \; ; \; \ketbra{\psi_{\alpha\beta}}{\psi_{\alpha\beta}} = P\left(\colvec{\vec{\sigma}_3^{AB}}{\vec{\sigma}_1^{AB}}, \colvec{\alpha}{\beta}\right)
	\end{equation}
	The first relation of Lemma \ref{lemma:SumOfBellStateProjections} now follows from
	\begin{equation}
	\rowvec{I}{I}\colvec{\vec{\sigma}_3^A}{\vec{\sigma}_3^B} = \rowvec{I}{0} \colvec{\vec{\sigma}_3^{AB}}{\vec{\sigma}_1^{AB}}
	\end{equation} 
	and Lemma \ref{lemma:RelationOfMeasurementProjectionsForTwoLinearlyRelatedTuples}. The second relation follows similarly. 
\end{proof}

\subsection{Approximately computing certain functions from only a two-universal hash of the input}\label{sec:ComputingCertainFunctions}

Take any subset $S \subset \F_2^n$. Consider the function $ f_S : \F_2^n  \rightarrow S \cup \{\bot\}$ given by  
\begin{equation}
f_S(\alpha) = \begin{cases}
\alpha & \text{if } \alpha \in S \\
\bot & \text{otherwise}
\end{cases}
\end{equation}
If $\alpha$ specifies errors, then $f_S$ computes whether $\alpha$ belongs to a set $S$ of acceptable errors, if so computes the entire string $\alpha$, and otherwise outputs an error message. It is very convenient to have functions of this form when constructing QKD protocols and security proofs.

It turns out that it is possible to approximately compute $f_S(\alpha)$ given only a two universal hash of the input. Recall \cite{carter1979universal,wegman1981new}:

\begin{definition}\label{def:TwoUniversalHash}
	A family of functions $\mathbf{H}$ from finite set $\mathbf{X}$ to finite set $\mathbf{Y}$ is two-universal with collision probability at most $\epsilon$ if for all $x \neq x' \in \mathbf{X}$, \begin{equation} \Pr_{h \leftarrow \mathbf{H}}(h(x)=h(x')) \leq \epsilon \end{equation} where the probability is taken over $h$ chosen uniformly from $\mathbf{H}$. If no explicit value is specified for the collision probability bound, then the default value $\epsilon = 1/|\mathbf{Y}|$ is taken. 
\end{definition}

Now, let $\mathbf{H}$ be a two-universal family from $\F_2^n$ to some finite set $\mathbf{Y}$ with collision probability bound $\epsilon$. Let $S = \{s_1, \dots, s_m\} $. Consider the function $g_S: \mathbf{H} \times \mathbf{Y} \rightarrow S \cup \{\bot\}$ given by the deterministic algorithm: 
\begin{enumerate}
	\item On input $h,y$, 
	\item For $i=1, \dots, m$, if $h(s_i) = y$, output $s_i$ and stop. 
	\item Output $\bot$. 
\end{enumerate}
Then:

\begin{theorem}\label{thm:ApproximateMembershipWithGeneralTwoUniversalHash}
	For all $n \in \mathbb{N}$, for all $\epsilon$, for all two-universal families $\mathbf{H}: \F_2^n \rightarrow \mathbf{Y}$ with collision probability bound $\epsilon$, for all subsets $S \subset \F_2^n$, for all $\alpha \in \F_2^n$,
	\begin{equation}
	\Pr_{h \leftarrow \mathbf{H}} (f_S(\alpha) \neq g_S(h,h(\alpha))) \leq \epsilon |S|
	\end{equation}
\end{theorem}

\begin{proof}
	The event \begin{equation} f_S(\alpha) \neq g_S(h,h(\alpha)) \end{equation} implies the event \begin{equation} \exists s \in S \backslash \{\alpha\} : h(s) = h(\alpha) \end{equation} The union bound and Definition \ref{def:TwoUniversalHash} give
	\begin{equation}
	\Pr_{h \leftarrow \mathbf{H}} (f_S(\alpha) \neq g_S(h,h(\alpha))) \leq \epsilon |S|
	\end{equation}
\end{proof}

The remainder of this section specializes Theorem \ref{thm:ApproximateMembershipWithGeneralTwoUniversalHash} to the case that the family $\mathbf{H}$ is a family of matrices over $\F_2$, and the set $S$ is a Hamming Ball. 

First, consider the following useful lemmas about random matrices over the field with two elements. 

Recall a property of random linear functions:
\begin{lemma}\label{lemma:RandomLinearFunctionIsTwoUniversal}
	Let $L$ be uniformly random in $\F_2^{k \times n}$, and take any fixed $x \in \F_2^n - \{0\}$. Then, $\Pr_L(Lx=0) = 2^{-k}$. 
\end{lemma}

\begin{proof}
	Take $i$ such that $x_i = 1$. Then, $Lx = L_i + L_{-i} x_{-i}$, where $L_i$ is the $i$-th column of $L$ and where $L_{-i},x_{-i}$ are formed from $L,x$ by omitting the $i$-th column and $i$-th entry respectively. Now, $L_i$ is uniform over $\F_2^k$ and independent from $L_{-i}$, so $Lx$ is also uniform over $\F_2^k$. 
\end{proof}

Thus, for all $y \neq z \in \F_2^n$, $\Pr_L(Ly=Lz) = 2^{-k}$, so random linear functions are two-universal. 

Later on, it will be more convenient to select matrices not from all of $\F_2^{k \times n}$, but from the subset consisting of those matrices of rank $k$. This subset also satisfies the two-universal condition, as the following two lemmas show. 

\begin{lemma}\label{lemma:NumberOfFullRankMatrices}
	For all integers $n \geq k \geq 1$, the number of rank $k$ matrices in $\F_2^{k \times n}$ is $ \prod_{i=1}^k (2^n - 2^{i-1}) $
\end{lemma}

\begin{proof}
	Given $i-1$ linearly independent rows, there are $2^{n} - 2^{i-1}$ ways to choose the $i$-th row outside their span. 
\end{proof}

\begin{lemma}\label{lemma:RandomFullRankMatrixIsTwoUniversal}
	Take $k \leq n$, let $L$ be a uniformly random rank $k$ matrix in $\F_2^{k \times n}$ and take any $x \in \F_2^n - \{0\}$. Then $\Pr_L(Lx=0) = \frac{2^{n-k}-1}{2^n-1} < 2^{-k}$
\end{lemma}

\begin{proof}
	Take invertible $M \in \F_2^{n \times n}$ such that $Mx = (1,0,\dots, 0)^T$. Then $\Pr(Lx=0)= \Pr(L M^{-1} Mx = 0)$. Now, find the probability that the first column of $LM^{-1}$ is zero. Note that $LM^{-1}$ is also uniformly distributed over the rank $k$ matrices in $\F_2^{k \times n}$, so the probability its first column is zero is the number of rank $k$ matrices in $\F_2^{k \times (n-1)}$ divided by the number of rank $k$ matrices in $\F_2^{k \times n}$. Lemma \ref{lemma:NumberOfFullRankMatrices} implies: \begin{equation} \Pr(L M^{-1} Mx = 0) = \frac{\prod_{i=1}^k (2^{n-1} - 2^{i-1})}{\prod_{i=1}^k (2^n - 2^{i-1})} = \frac{2^{n-k}-1}{2^n-1} < 2^{-k} \end{equation} completing the proof of Lemma \ref{lemma:RandomFullRankMatrixIsTwoUniversal}. 
\end{proof}

Interestingly, the collision probability bound $\epsilon = \frac{2^{n-k}-1}{2^n-1}$ achieved by the full rank matrices is the lowest possible for a two-universal family $\F_2^n \rightarrow \F_2^k$. This follows from a slight strengthening of \cite[Proposition 1]{carter1979universal}: 

\begin{lemma}\label{lemma:LowestParameterForTwoUniversalFamily}
	For every family $\mathbf{H}$ (not necessarily two-universal) of functions from finite set $\mathbf{X}$ to finite set $\mathbf{Y}$, there exist $x \neq x' \in \mathbf{X}$ such that \begin{equation} \Pr_{h \leftarrow \mathbf{H}}(h(x) = h(x')) \geq \frac{\frac{|\mathbf{X}|}{|\mathbf{Y}|} -1}{|\mathbf{X}|-1} \end{equation}
\end{lemma}

\begin{proof}
	Follow the same proof as \cite{carter1979universal} until the point they apply the pigeonhole principle. At that point, observe that the number of non-zero terms in the sum is not only less than $|\mathbf{X}|^2$, as they say there, but is in fact at most $|\mathbf{X}|(|\mathbf{X}|-1)$. 
	
	In more detail, for $h \in \mathbf{H}, x,x'\in \mathbf{X}$, define \begin{equation} \delta_h(x,x') = \begin{cases}
	1 & \text{if } x\neq x' \wedge h(x)=h(x') \\
	0 & \text{otherwise}
	\end{cases} \end{equation}
	For every $h \in \mathbf{H}$ partition $\mathbf{X} = \cup_{y \in \mathbf{Y}} h^{-1}(y)$ then observe that
	\begin{equation} \sum_{x,x'\in\mathbf{X}} \delta_h(x,x') = \sum_{y \in \mathbf{Y}} |h^{-1}(y)|(|h^{-1}(y)| - 1) \geq \frac{|\mathbf{X}|^2}{|\mathbf{Y}|} - |\mathbf{X}| \end{equation} by the quadratic mean-arithmetic mean inequality. Now, sum over $h \in \mathbf{H}$:
	\begin{equation} \sum_{h \in \mathbf{H}}\sum_{x,x'\in\mathbf{X}} \delta_h(x,x') = \sum_{x,x'\in\mathbf{X}} \sum_{h \in \mathbf{H}} \delta_h(x,x') \geq |\mathbf{H}| (\frac{|\mathbf{X}|^2}{|\mathbf{Y}|} - |\mathbf{X}|) \end{equation}
	Now, $ \sum_{h \in \mathbf{H}} \delta_h(x,x') $ is non-zero only when $x \neq x'$. Then, there exist $x \neq x'$ such that \begin{equation} \sum_{h \in \mathbf{H}} \delta_h(x,x') \geq |\mathbf{H}|\frac{\frac{|\mathbf{X}|}{|\mathbf{Y}|} -1}{|\mathbf{X}|-1} \end{equation}
\end{proof}

Later results will also use the fact that a row submatrix of a random invertible matrix has the uniform distribution over full rank matrices:

\begin{lemma}\label{lemma:RowSubmatrixOfRandomInvertibleMatrixIsRandomFullRankMatrix}
	Take any integers $n\geq k \geq 1$, and any $S \subset \{1,\dots,n\}$ of size $k$. Let $L$ be uniformly distributed over invertible matrices in $\F_2^{n \times n}$. Let $L_S$ denote the matrix formed by rows of $L$ with indices in $S$. Then, $L_S$ is uniformly distributed over full rank matrices in $\F_2^{k \times n}$. 
\end{lemma}

\begin{proof}
	Pick any fixed full rank $\Lambda \in \F_2^{k \times n}$. Compute $\Pr(L_S=\Lambda)$ as the number of ways to choose the remaining rows of $L$, which is $\prod_{i=1}^{n-k} (2^n - 2^{k+i-1})$	divided by the number of invertible matrices in $\F_2^{n \times n}$, which is $\prod_{i=1}^n (2^n - 2^{i-1}) $. Thus, 
	\begin{equation}
	\Pr(L_S=\Lambda) = \frac{\prod_{i=1}^{n-k} (2^n - 2^{k+i-1})}{\prod_{i=1}^n (2^n - 2^{i-1})} = \frac{1}{\prod_{i=1}^k (2^n - 2^{i-1})}
	\end{equation}
	Thus, $L_S$ is uniform over the full rank matrices in $\F_2^{k \times n}$. 
\end{proof}

Applying Theorem \ref{thm:ApproximateMembershipWithGeneralTwoUniversalHash} when the set $S$ is a Hamming ball requires a bound on the size of Hamming balls. For $x,y \in \F_2^n$, let $d_H(x,y) = |\{i:x_i \neq y_i\}|$ denote the Hamming distance between them. Let $B_n(x,r)$ denote the Hamming ball of radius $r$ around $x$. Then:

\begin{lemma}\label{lemma:HammingBallSize}
	For all $n,r \in \mathbb{N}$ such that $2r \leq n $, for all $x \in \F_2^n$, $|B_n(x,r)| < 2^{nh(r/n)}$
\end{lemma}

\begin{proof}
	\begin{multline}
	|B_n(x,r)| 2^{-nh(r/n)} = \sum_{i=0}^r {n \choose i} \left(\frac{r}{n}\right)^r \left(\frac{n-r}{n}\right)^{n-r} \\
	\leq \sum_{i=0}^r {n \choose i} \left(\frac{r}{n}\right)^i \left(\frac{n-r}{n}\right)^{n-i} < \sum_{i=0}^n {n \choose i} \left(\frac{r}{n}\right)^i \left(\frac{n-r}{n}\right)^{n-i} = 1
	\end{multline}
\end{proof}

From Theorem \ref{thm:ApproximateMembershipWithGeneralTwoUniversalHash}, Lemma \ref{lemma:RandomFullRankMatrixIsTwoUniversal} and Lemma \ref{lemma:HammingBallSize} deduce:

\begin{corollary}\label{cor:ApproximateMembershipWithRandomFullRankMatrixAndHammingBall}
	For all $n,k,r \in \mathbb{N}$ with $2r \leq n$ and $k \leq n$, for all $\alpha \in \F_2^n$, \begin{equation} \Pr_L(\functionf(\alpha) \neq \functiong(L,L\alpha)) < 2^{-k + nh(r/n)} \end{equation} where $L$ is chosen uniformly from the full rank matrices in $\F_2^{k\times n}$. 
\end{corollary}

\section{The two-universal hashing QKD protocol and its security}\label{sec:QKDProtocolAndItsSecurity}

Consider the following family $\pi(n,k,r)$ of entanglement-based QKD protocols, parameterized by $n,k,r \in \mathbb{N}$. The interpretation of the parameters is the following: $n$ is the number of qubits that each of Alice and Bob receive, $k$ is the size of each of their syndrome measurements and $n-2k$ is the size of their output secret key, and $r$ is the maximum number of bit flip or phase flip errors on which the protocol does not abort. The protocols output a secret key with security guarantees when $2n h(r/n) <2k < n$. 

It will be clear throughout that the size of the two syndrome measurements can vary independently, and so can the maximum number of tolerated bit flip and phase flip errors, but that would lead to overly complex notation, with five parameters $n,k,k',r,r'$, so it is not pursued explicitly below. 

\begin{enumerate}
	\item Alice and Bob each receive an $n$ qubit state from Eve, and they inform each other that the states have been received. 
	\item Alice and Bob publicly choose a random invertible $L \in \F_2^{n \times n}$. Let $L_1, L_2, L_3$ be the matrices formed by the first $k$ rows, the second $k$ rows, and the last $n-2k$ rows of $L$.  Let $M= (L^{-1})^T$, and let $M_1, M_2, M_3$ be the matrices formed by the first $k$, second $k$, and last $n-2k$ rows of $M$. $L_1$, $M_2$ are the parity check matrices of a CSS code. $L_3,M_3$ contain information about the logical $Z$ and $X$ operators on the codespace. 
	\item Alice applies the isometry $ \sum_{z} \ket{z,L_1 z}_{AU'_A}\bra{z}_A$ and Bob applies the isometry $\sum_{z} \ket{z, L_1z}_{BU'_B}\bra{z}_B$. This can be done by preparing $k$ ancilla qubits in state $0$ and applying a CNOT gate for each entry $L_1(i,j)$ that equals 1.  
	\item Alice and Bob measure all qubits in registers $A,B$ in the $\ket{+},\ket{-}$ basis, obtaining outcomes $x_A, x_B$. Alice and Bob measure all qubits in registers $U'_A,U'_B$ in the computational basis, obtaining outcomes $u_A,u_B$. 
	\item Alice and Bob compute $v_A = M_2 x_A$, $v_B=M_2 x_b$, $w_A=M_3 x_A$, $w_B=M_3 x_B$. 
	\item Alice and Bob discard registers $A,B,U'_A,U'_B$.
	\item Alice and Bob discard $x_A, x_B$, keeping only $v_A,v_B,w_A,w_B$. Thus, in effect, Alice and Bob erase $M_1 x_A, M_1x_B$. Note that the post measurement states in registers $A,B$, as well as $x_A,x_B$ have to be discarded in such a way that Eve cannot get them. 
	\item Alice and Bob announce $u_A, u_B, v_A, v_B$. Alice and Bob compute $s=g_{B_n(0,r)}(L_1, u_A+u_B)$ and $t=g_{B_n(0,r)}(M_2,v_A+v_B)$. 
	\item If both of $s,t$ are not $\bot$, then Alice takes $w_A$ to be the output secret key, and Bob takes $w_B + M_3 t$ to be the output secret key. 	
\end{enumerate}

As is usual in the literature on QKD, the protocol assumes that classical communication takes place over an authenticated channel. Unconditionally secure message authentication with composable security in the Abstract Cryptography framework can be obtained from a short secret key \cite{portmann2014key}, or using an advantage in channel noise \cite{ostrev2019composable}. 

If it is desired that the classical communication is minimized, then the following exchange of messages suffices: Bob confirms to Alice that he has received the qubits, Alice sends to Bob $L,u_A,v_A$, Bob informs Alice whether both of $s,t$ are not $\bot$. However, the initial formulation above better emphasizes the symmetry of the protocol, and makes clear that it is not important to keep the values $u_B,v_B,s,t$ secret. 

The following theorem establishes the security and robustness of the protocols $\pi(n,k,r)$. 

\begin{theorem}\label{thm:SecurityOfTheQKDProtocol}
	Take any $n,k,r \in\mathbf{N}$ such that $2n h(r/n) <2k < n$. Then, the protocol $\pi(n,k,r)$ is $2^{-k/2 + n h(r/n)/2 + 5/2 }$ secure. 

	Moreover, for any input state $\rho_{AB}$, the probability that $\pi(n,k,r)$ accepts on input $\rho_{AB}$ is $2^{-k/2 + n h(r/n)/2 + 3/2 }$ close to $Tr(\Pi_{n,r} \rho_{AB} \Pi_{n,r})$, where $\Pi_{n,r}$ is the projection on the subspace of systems $AB$ spanned by the Bell states with at most $r$ bit flip and at most $r$ phase flip errors.
\end{theorem}

\subsection{Proof of Theorem \ref{thm:SecurityOfTheQKDProtocol}}\label{subsec:ProofOfTheMainTheoremOnQKDSecurity}

The main idea of the proof of Theorem \ref{thm:SecurityOfTheQKDProtocol} is that the real values $g_{B_n(0,r)}(L_1, u_A+u_B)$ and $g_{B_n(0,r)}(M_2,v_A+v_B)$ computed during the protocol can be replaced by the corresponding ideal values $\functionf(\alpha),\functionf(\beta)$. From now on, use shorthand notation and skip the subscript $B_n(0,r)$, thus writing $f$ for $\functionf$ and $g$ for $\functiong$. 

The steps of the proof of Theorem \ref{thm:SecurityOfTheQKDProtocol} are the propositions below. Start by writing the action of the protocol as an isometry followed by a partial trace. 

\begin{proposition}\label{prop:TheStateAtTheEndOfTheProtocol}
	Let $\mathcal{E}_{real}$ be the completely positive trace preserving transformation applied by the first eight steps of the protocol. Then, for all input states $\rho_{ABE}$ to the protocol, the output state $\mathcal{E}_{real}(\rho_{ABE})$ of the classical registers $\mathbf{L},U_A,U_B,V_A,V_B,W_A,W_B,S,T$ and the quantum register of Eve equals 
	\begin{equation}
Tr_{AB\mathbf{L'}S'T'U'_AU'_BV'_AV'_BW'_AW'_B}  \mathcal{W} \mathcal{V}_{real} \mathcal{U}_{real} \left(\rho \otimes \ketbra{\mathcal{L}}{\mathcal{L}}\right) \mathcal{U}_{real}^\dagger \mathcal{V}_{real}^\dagger \mathcal{W}^\dagger 
\end{equation}
where 
\begin{equation}
\ket{\mathcal{L}} = \sum_{L} \sqrt{p_L} \ket{LL}_{\mathbf{LL'}}
\end{equation}
is a purification of the choice of random matrix $L$, where 
\begin{multline}
\mathcal{U}_{Real} = \sum_{L,z_A,z_B} \ketbra{L}{L}_{\mathbf{L}} \otimes \ketbra{z_Az_B}{z_Az_B}_{AB} \\ \otimes \ket{L_1z_A,L_1z_A,L_1z_B,L_1z_B,g(L_1,L_1(z_A+z_B)),g(L_1,L_1(z_A+z_B))}_{U_AU'_AU_BU'_BSS'}
\end{multline} 
is an isometry that captures the measurement through which Alice and Bob obtain the values $u_A=L_1z_A$ and $u_B=L_1z_B$ as well as the subsequent computation of the value $s=g(L_1, L_1(z_A+z_B))$, where 
\begin{multline}
\mathcal{V}_{Real} = \sum_{L,x_A,x_B} \ketbra{L}{L}_{\mathbf{L}} \otimes \left( H^{\otimes 2n} \ketbra{x_Ax_B}{x_Ax_B} H^{\otimes 2n}\right)_{AB} \\ \otimes \ket{M_2x_A,M_2x_A,M_2x_B,M_2x_B,g(M_2,M_2(x_A+x_B)),g(M_2,M_2(x_A+x_B))}_{V_AV'_AV_BV'_BTT'}
\end{multline} 
is an isometry that captures the measurement through which Alice and Bob obtain the values $v_A=M_2x_A$ and $v_B=M_2x_B$ as well as the subsequent computation of the value $t=g(M_2, M_2(x_A+x_B))$ and where 
\begin{multline}
\mathcal{W} = \sum_{L,x_A,x_B} \ketbra{L}{L}_{\mathbf{L}} \otimes \left( H^{\otimes 2n} \ketbra{x_Ax_B}{x_Ax_B} H^{\otimes 2n}\right)_{AB} \\ \otimes \ket{M_3x_A,M_3x_A,M_3x_B,M_3x_B}_{W_AW'_AW_BW'_B}
\end{multline} 
is an isometry that captures the measurement through which Alice and Bob obtain the values $w_A=M_3x_A$ and $w_B=M_3x_B$. 
\end{proposition}

\begin{proof}
Recall the Stinespring dilation theorem \cite{stinespring1955positive}. Systematically express each step of the protocol as an isometry followed by a partial trace. 

The step in which Alice and Bob choose the random matrix $L$ can be expressed as preparing the purification $\ket{\mathcal{L}}_{\mathbf{LL'}}$
and then taking $Tr_\mathbf{L'}$. 

The steps in which Alice and Bob apply the isometry 
\begin{equation}
\sum_{z_A,z_B} \ket{z_A,z_B,L_1z_A,L_1z_B}_{ABU'_AU'_B} \bra{z_A,z_B}_{AB}
\end{equation}
then measure registers $U'_A,U'_B$ in the computational basis, discarding the post-measurement state and keeping only the outcome, then compute the value $s$ can be expressed by the isometry $\mathcal{U}_{real}$ followed by $Tr_{S'U'_AU'_B}$. 

The steps in which Alice and Bob measure the qubits in $A,B$ in the $\ket{+}, \ket{-}$ basis obtaining $x_A,x_B$, then compute $v_A,v_B,w_A,w_B,t$, then discard the post-measurement state of the qubits in $A,B$ and the outcomes $x_A,x_B$ can be expressed by the product of isometries $\mathcal{W}\mathcal{V}_{real}$ followed by $Tr_{ABT'V'_AV'_BW'_AW'_B}$.

Finally, note that all the partial trace operations can be commuted to the end.  
\end{proof}

Next, note that $\mathcal{U}_{real}$ can be approximated by an ideal isometry followed by a simulator isometry. 

\begin{proposition}\label{prop:RealAndIdealIsometriesForBitFlipErrors}
	Let
	\begin{equation}
	\mathcal{U}_{ideal} = \sum_{\alpha,\beta} \ketbra{\psi_{\alpha\beta}}{\psi_{\alpha\beta}}_{AB} \otimes \ket{f(\alpha),f(\alpha)}_{SS'}
	\end{equation}
	This ideal isometry computes whether the number of bit flip errors is acceptable and if so it computes the entire string of bit flip error positions. 
	
	Let
	\begin{equation}
	\mathcal{U}_{simulator} = \sum_{L,z_A,z_B} \ketbra{L}{L}_{\mathbf{L}} \otimes \ketbra{z_Az_B}{z_Az_B}_{AB} \\ \otimes \ket{L_1z_A,L_1z_A,L_1z_B,L_1z_B}_{U_AU'_AU_BU'_B}
	\end{equation}
	This isometry captures the measurement through which Alice and Bob obtain the values $u_A=L_1z_A$ and $u_B=L_1z_B$. 
	
	Then:
	\begin{equation}
	\left( \bra{\mathcal{L}} \mathcal{U}_{ideal}^\dagger \mathcal{U}_{simulator}^\dagger \right) \left(\mathcal{U}_{real} \ket{\mathcal{L}}\right) \geq (1- 2^{-k+ n h(r/n)}) I_{AB}
	\end{equation}
\end{proposition}

\begin{proof}
	Simplify:
	\begin{multline}
	\mathcal{U}_{simulator}^\dagger  \mathcal{U}_{real}  
	= \sum_{L,z_A,z_B} \ketbra{L}{L}_{\mathbf{L}} \otimes \ketbra{z_Az_B}{z_Az_B}_{AB} \\ \otimes \ket{g(L_1,L_1(z_A+z_B)),g(L_1,L_1(z_A+z_B))}_{SS'}
	\end{multline}
	Therefore,
	\begin{multline}
	\left( \bra{\mathcal{L}} \mathcal{U}_{ideal}^\dagger \mathcal{U}_{simulator}^\dagger \right) \left(\mathcal{U}_{real} \ket{\mathcal{L}}\right) \\
	= \sum_{L,z_A,z_B,\alpha,\beta} p_L \ketbra{\psi_{\alpha\beta}}{\psi_{\alpha\beta}}_{AB} \ketbra{z_Az_B}{z_Az_B}_{AB} \braket{f(\alpha)}{g(L_1,L_1(z_A+z_B))}_S
	\end{multline}
	Now, apply Lemma \ref{lemma:SumOfBellStateProjections}:
	\begin{multline}
	\sum_{L,z_A,z_B,\alpha} p_L \left( \sum_{\beta}\ketbra{\psi_{\alpha\beta}}{\psi_{\alpha\beta}}_{AB}\right) \ketbra{z_Az_B}{z_Az_B}_{AB} \braket{f(\alpha)}{g(L_1,L_1(z_A+z_B))}_S \\
	=\sum_{L,z_A,z_B,\alpha,z'_A} p_L \ketbra{z'_A,z'_A+\alpha}{z'_A,z'_A+\alpha}_{AB} \ketbra{z_Az_B}{z_Az_B}_{AB} \braket{f(\alpha)}{g(L_1,L_1(z_A+z_B))}_S \\
	= \sum_{z_A,z_B} \ketbra{z_Az_B}{z_Az_B}_{AB} \sum_L p_L \braket{f(z_A+z_B)}{g(L_1,L_1(z_A+z_B))} \\
	= \sum_{z_A,z_B} \ketbra{z_Az_B}{z_Az_B}_{AB} \Pr_L(f(z_A+z_B) = g(L_1,L_1(z_A+z_B)))
	\end{multline}
	Now, the marginal distribution of $L_1$ is uniform over the rank $k$ matrices in $\F_2^{k \times n}$ because $L$ is selected uniformly among invertible matrices in $\F_2^{n \times n}$ (Lemma \ref{lemma:RowSubmatrixOfRandomInvertibleMatrixIsRandomFullRankMatrix}). Complete the proof of Proposition \ref{prop:RealAndIdealIsometriesForBitFlipErrors} by applying Corollary \ref{cor:ApproximateMembershipWithRandomFullRankMatrixAndHammingBall}. 
\end{proof}

Next, perform the same approximation for $\mathcal{V}_{real}$. 

\begin{proposition}\label{prop:RealAndIdealIsometriesForPhaseFlipErrors}
	Let
	\begin{equation}
	\mathcal{V}_{ideal} = \sum_{\alpha,\beta} \ketbra{\psi_{\alpha\beta}}{\psi_{\alpha\beta}}_{AB} \otimes \ket{f(\beta),f(\beta)}_{TT'}
	\end{equation}
	This ideal isometry computes whether the number of phase flip errors is acceptable and if so it computes the entire string of phase flip error positions. 
	
	Let
	\begin{multline}
	\mathcal{V}_{simulator} = \sum_{L,x_A,x_B} \ketbra{L}{L}_{\mathbf{L}} \otimes \left( H^{\otimes 2n} \ketbra{x_Ax_B}{x_Ax_B} H^{\otimes 2n}\right)_{AB} \\ \otimes \ket{M_2x_A,M_2x_A,M_2x_B,M_2x_B}_{V_AV'_AV_BV'_B}
	\end{multline}
	This isometry captures the measurement through which Alice and Bob obtain the values $v_A=M_2x_A$ and $v_B=M_2x_B$. 
	
	Then:
	\begin{equation}
	\left( \bra{\mathcal{L}} \mathcal{V}_{ideal}^\dagger  \mathcal{V}_{simulator}^\dagger \right) \left(\mathcal{V}_{real} \ket{\mathcal{L}}\right) \geq (1- 2^{-k+ n h(r/n)}) I_{AB}
	\end{equation}
\end{proposition}

\begin{proof}	
As in the proof of Proposition \ref{prop:RealAndIdealIsometriesForBitFlipErrors}, use Lemma \ref{lemma:SumOfBellStateProjections} to compute
	\begin{multline}
	\left( \bra{\mathcal{L}} \mathcal{V}_{ideal}^\dagger  \mathcal{V}_{simulator}^\dagger \right) \left(\mathcal{V}_{real} \ket{\mathcal{L}}\right) \\
	= \sum_{x_A,x_B} \left( H^{\otimes 2n} \ketbra{x_Ax_B}{x_Ax_B} H^{\otimes 2n}\right)_{AB} \Pr_L(f(x_A+x_B) = g(M_2,M_2(x_A+x_B)))
	\end{multline}
Now, $M = (L^{-1})^T$ is uniformly distributed over invertible matrices in $\F_2^{n \times n}$, so Lemma \ref{lemma:RowSubmatrixOfRandomInvertibleMatrixIsRandomFullRankMatrix} and Corollary \ref{cor:ApproximateMembershipWithRandomFullRankMatrixAndHammingBall} complete the proof. 
\end{proof}

Next, observe that:

\begin{proposition}\label{prop:CommutingIsometries}
$\mathcal{U}_{simulator}\mathcal{V}_{real} = \mathcal{V}_{real} \mathcal{U}_{simulator}$
\end{proposition}

\begin{proof}	
	Rewrite:
	\begin{multline}
	\mathcal{U}_{simulator} = \sum_{L,z_A,z_B} \ketbra{L}{L}_{\mathbf{L}} \otimes \ketbra{z_Az_B}{z_Az_B}_{AB}  \otimes \ket{L_1z_A,L_1z_A,L_1z_B,L_1z_B}_{U_AU'_AU_BU'_B} \\
	= \sum_{L,u_A,u_B} \ketbra{L}{L}_{\mathbf{L}} \otimes \ket{u_A,u_A,u_B,u_B}_{U_AU'_AU_BU'_B} \\
	\otimes \left(\sum_{z_A: L_1z_A=u_A} \ketbra{z_A}{z_A}\right)_A \otimes \left(\sum_{z_B: L_1z_B=u_B} \ketbra{z_B}{z_B}\right)_B \\
	=  \sum_{L,u_A,u_B} \ketbra{L}{L}_{\mathbf{L}} \otimes \ket{u_A,u_A,u_B,u_B}_{U_AU'_AU_BU'_B} 
	\otimes P(L_1(\vec{\sigma}_3),u_A)_A \otimes P(L_1(\vec{\sigma}_3),u_B)_B
	\end{multline}
	where the last step uses Lemma \ref{lemma:RelationOfMeasurementProjectionsForTwoLinearlyRelatedTuples} and the notation of Section \ref{subsec:ThePauliGroupAndTheBellBasis} for the tuple $\vec{\sigma}_3$ of single qubit $\sigma_3$ operations. 
	
	Similarly, rewrite
	\begin{multline}
	\mathcal{V}_{Real} = \sum_{L,x_A,x_B} \ketbra{L}{L}_{\mathbf{L}} \otimes \left( H^{\otimes 2n} \ketbra{x_Ax_B}{x_Ax_B} H^{\otimes 2n}\right)_{AB} \\ \otimes \ket{M_2x_A,M_2x_A,M_2x_B,M_2x_B,g(M_2,M_2(x_A+x_B)),g(M_2,M_2(x_A+x_B))}_{V_AV'_AV_BV'_BTT'} \\
	= \sum_{L,v_A,v_B} \ketbra{L}{L}_{\mathbf{L}} \otimes P(M_2(\vec{\sigma}_1),v_A)_A \otimes P(M_2(\vec{\sigma}_1),v_B)_B \\ \otimes \ket{v_A,v_A,v_B,v_B,g(M_2,v_A+v_B),g(M_2,v_A+v_B}_{V_AV'_AV_BV'_BTT'} 
	\end{multline}
	where $\vec{\sigma}_1$ is the tuple of single qubit $\sigma_1$ operations. 
	
	Proposition \ref{prop:CommutingIsometries} now follows by observing that the elements of the two tuples $L_1(\vec{\sigma}_3)$ and $M_2(\vec{\sigma}_1)$ commute and therefore for all $u,v$, the corresponding projections $P(L_1(\vec{\sigma}_3),u)$ and $P(M_2(\vec{\sigma}_1),v)$ also commute. 
\end{proof}

Next, use propositions \ref{prop:TheStateAtTheEndOfTheProtocol}, \ref{prop:RealAndIdealIsometriesForBitFlipErrors}, \ref{prop:RealAndIdealIsometriesForPhaseFlipErrors}, \ref{prop:CommutingIsometries} to construct an ideal transformation that approximates $\mathcal{E}_{real}$:

\begin{proposition}\label{prop:IdealCPTPmap}
	Let $\mathcal{E}_{ideal}$ be the transformation that prepares $\ket{\mathcal{L}}$, then applies isometries $\mathcal{U}_{ideal}$, $\mathcal{V}_{ideal}$, $\mathcal{V}_{simulator}$, $\mathcal{U}_{simulator}$, $\mathcal{W}$, and finally applies $Tr_{AB\mathbf{L'}S'T'U'_AU'_BV'_AV'_BW'_AW'_B}$. Then, the diamond distance of $\mathcal{E}_{real}$ and $\mathcal{E}_{ideal}$ is at most $2^{-k/2 + n h(r/n) / 2 + 3/2}$.
\end{proposition}

\begin{proof}
	Take any input state $\rho_{ABE}$ and purify it to $\ket{\phi}_{ABEE'}$. From Proposition \ref{prop:RealAndIdealIsometriesForBitFlipErrors} deduce that the fidelity of  $\mathcal{V}_{real} \mathcal{U}_{simulator}\mathcal{U}_{ideal} \ket{\phi} \ket{\mathcal{L}}$ and $ \mathcal{V}_{real} \mathcal{U}_{real} \ket{\phi} \ket{\mathcal{L}}$ is at least $1- 2^{-k+nh(r/n)}$. Using the relation of fidelity and trace distance for pure states \cite[Equation 9.99]{nielsen2012quantum}, the trace distance between these two states is 
	\begin{equation}
	\sqrt{1 - (1 - 2^{-k+nh(r/n)})^2} \leq 2^{-k/2 + n h(r/n) / 2 + 1/2}
	\end{equation}
	Next, from Proposition \ref{prop:CommutingIsometries} deduce
	\begin{equation}
	\mathcal{V}_{real} \mathcal{U}_{simulator}\mathcal{U}_{ideal} \ket{\phi} \ket{\mathcal{L}} =  \mathcal{U}_{simulator} \mathcal{V}_{real} \mathcal{U}_{ideal} \ket{\phi} \ket{\mathcal{L}}
	\end{equation}
	Next, from Proposition \ref{prop:RealAndIdealIsometriesForPhaseFlipErrors} deduce that the fidelity of 
$$ \mathcal{U}_{simulator} \mathcal{V}_{real} \mathcal{U}_{ideal} \ket{\phi} \ket{\mathcal{L}}$$
and 
$$\mathcal{U}_{simulator} \mathcal{V}_{simulator} \mathcal{V}_{ideal} \mathcal{U}_{ideal} \ket{\phi} \ket{\mathcal{L}} $$ 
is at least $1- 2^{-k+nh(r/n)}$, so the trace distance between them is at most $2^{-k/2 + n h(r/n) / 2 + 1/2}$. Finally, from Proposition \ref{prop:TheStateAtTheEndOfTheProtocol}, the triangle inequality and monotonicity of the trace distance deduce that the trace distance between $\mathcal{E}_{real}(\rho)$ and $\mathcal{E}_{ideal}(\rho)$ is at most $2^{-k/2 + n h(r/n) / 2 + 3/2}$.
\end{proof}

Next, compute the output state of $\mathcal{E}_{ideal}$: 

\begin{proposition}\label{prop:OutputOfIdealTransformation}
Take any input state $\rho_{ABE}$ and purify it to $\ket{\phi}_{ABEE'}$. Expand $\phi$ in the Bell basis for Alice and Bob:
\begin{equation}
\ket{\phi}_{ABEE'} = \sum_{\alpha, \beta \in \F_2^n} \ket{\psi_{\alpha\beta}}_{AB} \otimes \ket{\gamma_{\alpha\beta}}_{EE'}
\end{equation}
where $\ket{\gamma_{\alpha\beta}}$ are vectors in Eve's space that satisfy 
\begin{equation}
\sum_{\alpha, \beta \in \F_2^n} \braket{\gamma_{\alpha\beta}}{\gamma_{\alpha\beta}} = 1
\end{equation}
Then, there exists $\sigma^{reject}_{LEE'STU_AV_AW_AU_BV_BW_B}$ that is classical on registers $LSTU_AV_AW_AU_BV_BW_B$ and such that at least one of $ST$ contains $\bot$ and such that
\begin{multline}
\mathcal{E}_{ideal}(\ketbra{\phi}{\phi}) = \sigma^{reject}_{LEE'STU_AV_AW_AU_BV_BW_B}\\+ \sum_{L,u_A,v_A,w_A,\alpha,\beta: \alpha, \beta \in B_n(0,r)} p_L \ketbra{L}{L}_{\mathbf{L}}  \otimes \ketbra{\gamma_{\alpha\beta}}{\gamma_{\alpha\beta}}_{EE'}  \otimes \ketbra{\alpha,\beta}{\alpha,\beta}_{ST} \\
	\otimes 2^{-n} \ketbra{u_A,v_A,w_A}{u_A,v_A,w_A}_{U_AV_AW_A} \\
	\otimes \ketbra{u_A + L_1\alpha,v_A + M_2 \beta, w_A + M_3\beta}{u_A + L_1 \alpha,v_A + M_2 \beta, w_A + M_3\beta}_{U_BV_BW_B}
\end{multline}
\end{proposition}

\begin{proof}
	
	Simplify:
	\begin{equation}
	\mathcal{V}_{ideal} \mathcal{U}_{ideal} = \sum_{\alpha,\beta} \ketbra{\psi_{\alpha\beta}}{\psi_{\alpha\beta}} \otimes \ket{f(\alpha),f(\alpha),f(\beta),f(\beta)}_{SS'TT'}
	\end{equation}
	Also,  
	\begin{multline}
	\mathcal{W}\mathcal{V}_{simulator}\mathcal{U}_{simulator} = \sum_{L,u_A,u_B,v_A,v_B,w_A,w_B}  \ketbra{L}{L}_{\mathbf{L}} \\
	\otimes P\left(\colvecthree{L_1 \vec{\sigma}_3}{M_2 \vec{\sigma}_1}{M_3\vec{\sigma}_1},\colvecthree{u_A}{v_A}{w_A}\right)_A \otimes P\left(\colvecthree{L_1 \vec{\sigma}_3}{M_2 \vec{\sigma}_1}{M_3\vec{\sigma}_1},\colvecthree{u_B}{v_B}{w_B}\right)_B \\
	\otimes \ket{u_A,u_A,u_B,u_B,v_A,v_A,v_B,v_B,w_A,w_A,w_B,w_B}_{U_AU'_AU_BU'_BV_AV'_AV_BV'_BW_AW'_AW_BW'_B}
	\end{multline}
	using the notation of section \ref{subsec:ThePauliGroupAndTheBellBasis}, the observation that the elements of the three tuples $L_1 \vec{\sigma}_3, M_2 \vec{\sigma}_1, M_3 \vec{\sigma}_1$ are independent and commute, and Lemma \ref{lemma:RelationOfMeasurementProjectionsForTwoLinearlyRelatedTuples}. 
	
	Next, use Lemma \ref{lemma:ActionOfPauliMeasurementOnBellBasis} to deduce that 
	\begin{multline}
	\mathcal{W}\mathcal{V}_{simulator}\mathcal{U}_{simulator}\mathcal{V}_{ideal} \mathcal{U}_{ideal} \ket{\phi}\ket{\mathcal{L}} \\
	= \sum_{L,u_A,v_A,w_A,\alpha,\beta} \sqrt{p_L} \ket{L,L}_{\mathbf{LL'}} \otimes P\left(\colvecthree{L_1 \vec{\sigma}_3}{M_2 \vec{\sigma_1}}{M_3\vec{\sigma}_1},\colvecthree{u_A}{v_A}{w_A}\right)_A  \ket{\psi_{\alpha\beta}}_{AB} \otimes \ket{\gamma_{\alpha\beta}}_{EE'} \\
	\otimes \ket{u_A,u_A,u_A + L_1 \alpha,u_A + L_1 \alpha}_{U_AU'_AU_BU'_B} \\
	\otimes \ket{v_A,v_A,v_A + M_2\beta,v_A + M_2\beta}_{V_AV'_AV_BV'_B} \\
	\otimes \ket{w_A,w_A,w_A + M_3\beta,w_A + M_3\beta}_{W_AW'_AW_BW'_B} \\
	\otimes \ket{f(\alpha),f(\alpha),f(\beta),f(\beta)}_{SS'TT'}
	\end{multline}
	
	Next, break this up into a sum of two sub-normalized vectors $\ket{\tau_{accept}}$ and $\ket{\tau_{reject}}$, where $\ket{\tau_{accept}}$ contains those terms of the sum with $\alpha,\beta \in B_n(0,r)$ and $\ket{\tau_{reject}}$ contains all other terms of the sum. Note that $Tr_{S'T'} \ket{\tau_{accept}} \bra{\tau_{reject}}=0$ and deduce
	\begin{multline}
	\mathcal{E}^{ideal}(\ketbra{\phi}{\phi}) = Tr_{AB\mathbf{L'}S'T'U'_AU'_BV'_AV'_BW'_AW'_B} \ketbra{\tau_{accept}}{\tau_{accept}} \\
	+ Tr_{AB\mathbf{L'}S'T'U'_AU'_BV'_AV'_BW'_AW'_B} \ketbra{\tau_{reject}}{\tau_{reject}} 
	\end{multline}

Take \begin{equation} \sigma^{reject}_{LEE'STU_AV_AW_AU_BV_BW_B}= Tr_{AB\mathbf{L'}S'T'U'_AU'_BV'_AV'_BW'_AW'_B} \ketbra{\tau_{reject}}{\tau_{reject}} \end{equation}
	
	Finally, simplify and use Lemma \ref{lemma:PropertiesOfMaximallyEntangledState} to deduce that 
	\begin{multline}
	Tr_{AB\mathbf{L'}S'T'U'_AU'_BV'_AV'_BW'_AW'_B} \ketbra{\tau_{accept}}{\tau_{accept}} \\
	= \sum_{L,u_A,v_A,w_A,\alpha,\beta: \alpha, \beta \in B_n(0,r)} p_L \ketbra{L}{L}_{\mathbf{L}} \\
	\otimes \ketbra{\gamma_{\alpha\beta}}{\gamma_{\alpha\beta}}_{EE'} \left(\bra{\psi_{\alpha\beta}} P\left(\colvecthree{L_1 \vec{\sigma}_3}{M_2 \vec{\sigma_1}}{M_3\vec{\sigma}_1},\colvecthree{u_A}{v_A}{w_A}\right)_A \ket{\psi_{\alpha\beta}}\right) \\
	\otimes \ketbra{u_A,v_A,w_A}{u_A,v_A,w_A}_{U_AV_AW_A} \\
	\otimes \ketbra{u_A + L_1\alpha,v_A + M_2 \beta, w_A + M_3\beta}{u_A + L_1 \alpha,v_A + M_2 \beta, w_A + M_3\beta}_{U_BV_BW_B}\\
	\otimes \ketbra{\alpha,\beta}{\alpha,\beta}_{ST} \\
	= \sum_{L,u_A,v_A,w_A,\alpha,\beta: \alpha, \beta \in B_n(0,r)} p_L \ketbra{L}{L}_{\mathbf{L}}  \otimes \ketbra{\gamma_{\alpha\beta}}{\gamma_{\alpha\beta}}_{EE'}  \otimes \ketbra{\alpha,\beta}{\alpha,\beta}_{ST} \\
	\otimes 2^{-n} \ketbra{u_A,v_A,w_A}{u_A,v_A,w_A}_{U_AV_AW_A} \\
	\otimes \ketbra{u_A + L_1\alpha,v_A + M_2 \beta, w_A + M_3\beta}{u_A + L_1 \alpha,v_A + M_2 \beta, w_A + M_3\beta}_{U_BV_BW_B}
	\end{multline}
	which completes the proof. 	
\end{proof}

Finally, note that for any input state $\rho_{ABE}$, applying the final step of the protocol (the correction of $w_A,w_B$ depending on the values $s,t$) to $\mathcal{E}_{ideal}(\rho)$ produces an ideal state that satisfies the assumptions of Lemma \ref{lemma:RelaxingImplicitAssumptionsInEpsilonSecurity} with $\epsilon=2^{-k/2 + n h(r/n) / 2 + 3/2}$; therefore the protocol is $2^{-k/2 + n h(r/n) / 2 + 5/2}$ secure. Moreover, for any input $\rho_{ABE}=Tr_{E'} \ketbra{\phi}{\phi}_{ABEE'}$, the probability that the protocol accepts is within $2^{-k/2 + n h(r/n) / 2 + 3/2}$ of \begin{equation} \sum_{\alpha,\beta \in B_n(0,r)} \braket{\gamma_{\alpha\beta}}{\gamma_{\alpha\beta}} = Tr \Pi_{n,r} \rho_{AB} \Pi_{n,r} \end{equation} This completes the proof of Theorem \ref{thm:SecurityOfTheQKDProtocol}.

\section{Comparison with previous work}\label{sec:ComparisonWithPreviousWork}

The introduction illustrated the advantage of two-universal hashing over random sampling using specific examples. This section reveals the general pattern behind the examples in the introduction. To study the advantage of the two-universal hashing protocol for all block sizes, fix values for the tolerated error rate and security level, and consider key rate as a function of block size. How fast does key rate converge to the asymptotic value as block size goes to infinity? Subsection \ref{subsec:KeyRateFor2UH} gives the rate of convergence for the two-universal hashing protocol. Subsection \ref{subsec:KeyRateForRandomSampling} gives a bound on the rate of convergence of the random sampling protocol. 

\subsection{Key rate of the two-universal hashing protocols $\pi(n,k,r)$}\label{subsec:KeyRateFor2UH}

Given $n$ qubits per side, the target to tolerate $\delta n$ bit flip and $\delta n$ phase flip errors, and a target security parameter $\epsilon$, it suffices to choose $k=\lceil n h(\delta) + 2\log_2(1/\epsilon) + 5 \rceil $. The key rate $1-2k/n$ then satisfies: 
\begin{equation} 1-2h(\delta) - \frac{4 \log_2(1/\epsilon) + 12}{n} \leq 1-\frac{2k}{n} \leq 1-2h(\delta) - \frac{4 \log_2(1/\epsilon) + 10}{n} \end{equation}
Therefore, the rate of convergence of the finite to the asymptotic rate is of the form $cn^{-1}$. 

\subsection{Key rate of the random sampling protocols}\label{subsec:KeyRateForRandomSampling}

The sequence of works \cite{tomamichel2012tight,leverrier2017largely,lim2021security} develops QKD protocols and security proofs optimized for the finite key regime. The current evolution of the entanglement-based protocol can be found in \cite[Section 3]{leverrier2017largely}; the difference between \cite{lim2021security} and \cite{leverrier2017largely} is only in the random sampling tail bound that is used. For comparison with the present work we take only the case of perfect measurements in the rectilinear and diagonal basis. A summary of the protocol in this case is as follows:
\begin{enumerate}
\item Eve prepares a state of $2n$ qubits and sends $n$ to Alice and $n$ to Bob. 
\item Alice and Bob agree on a uniformly random choice of either the rectilinear or the diagonal basis measurement for each pair of qubits. 
\item Alice and Bob select a uniformly random subset of $n_{pe}$ positions to serve for parameter estimation, leaving the remaining $n_{rk}=n-n_{pe}$ to serve as the raw key. 
\item Alice and Bob compare their outcomes on the parameter estimation positions. If the error rate on these positions exceeds a threshold $\delta$, Alice and Bob abort. 
\item Alice sends a syndrome of her raw key to Bob, and a two-universal hash of her raw key to Bob. Bob uses the syndrome to correct his raw key, and uses the hash to verify that the correction was successful. For simplicity, take the combined length of syndrome and hash to be the theoretical minimum $n_{rk}h(\delta) - \log_2 (\epsilon_{ec})$, where $\epsilon_{ec}$ is the desired bound on the probability that the hash test passes but Bob's corrected raw key does not match Alice's. 
\item Alice and Bob compress their raw keys to shorter output keys of length $n_{out}$ using a two-universal family of hash functions. 
\end{enumerate}

The security $\epsilon_{qkd}$ of these protocols can be written in the form \begin{equation}\epsilon_{qkd} = \epsilon_{ec} + \inf_{0<\nu<1/2-\delta} \left(\epsilon_{pa}(\nu) + \epsilon_{pe}(\nu)\right) \end{equation} where $\epsilon_{ec}$ is the desired bound on the correctness of the protocol, where \begin{equation} \epsilon_{pa}(\nu) =\frac{1}{2\sqrt{\epsilon_{ec}}} 2^{(-n_{rk}(1-h(\delta+\nu)-h(\delta))+n_{out})/2}\end{equation} is a bound on the secrecy of the protocol, and where \begin{equation}  \epsilon_{pe}(\nu)=\inf_{0<\xi<\nu} \epsilon_{pe}(\nu,\xi) \end{equation} comes from a tail bound for random sampling. The precise form of the function $ \epsilon_{pe}(\nu,\xi) $ is given in \cite[Lemma 2]{lim2021security} and satisfies the equation \begin{equation} \left(\frac{\epsilon_{pe}(\nu,\xi)}{2}\right)^2= exp\left(-\frac{2nn_{pe}\xi^2}{n_{rk}+1}\right)+exp\left(-\frac{2(n+2)(n_{rk}^2(\nu-\xi)^2-1)}{(n(\delta+\xi)+1)(n(1-\delta-\xi)+1)}\right) \end{equation}
For the purpose of this section, consider the following lower bound on $\epsilon_{pe}(\nu)$:

\begin{lemma}\label{lemma:LowerBoundOnLimEtAlFunction}
Suppose $n_{rk} \geq n/2$. Then, \begin{equation} \epsilon_{pe}(\nu) \geq 2 exp(-2 n_{pe} \nu^2) \end{equation}
\end{lemma}

\begin{proof}

Take any $\xi \in (0,\nu)$. Note that \begin{equation} \frac{2nn_{pe}\xi^2}{n_{rk}+1} \leq 4 n_{pe} \nu^2 \end{equation} and therefore 
\begin{equation} exp\left(-\frac{2nn_{pe}\xi^2}{n_{rk}+1}\right) \geq exp\left(-4n_{pe}\nu^2\right) \end{equation} The lemma follows.
\end{proof}

The following bound holds on the key rate of the random sampling protocols:

\begin{theorem}\label{thm:UpperBoundOnKeyRateForRandomSampling}
Fix the block size $n$, the tolerated error rate $\delta$ and the security level \begin{equation} \epsilon_{qkd}= \epsilon_{ec} +  \inf_{0<\nu<1/2-\delta} \left(\epsilon_{pa}(\nu) + \epsilon_{pe}(\nu)\right) \end{equation} Then, the key rate $n_{out}/n$ is upper bounded by the larger of $(1-2h(\delta))/2$ and  \begin{equation} (1-2h(\delta)) - c_1(\epsilon_{qkd},\delta) n^{-1/3} - c_2(\epsilon_{qkd})n^{-1} \end{equation}
where
\begin{align}
c_1(\epsilon_{qkd},\delta) &= \frac{3}{2^{5/3}} (1-2h(\delta))^{1/3} \left(\frac{1-h(\delta)}{1/2 - \delta}\right)^{2/3} \left( \ln \frac{2}{\epsilon_{qkd}}\right)^{1/3}\\
c_2(\epsilon_{qkd}) &= 3 \log_2 (1/\epsilon_{qkd}) + 3\log_2(3) -4 
\end{align}
\end{theorem}

\begin{proof}
Take the optimal $\nu$. In case $n_{rk}/n < 1/2$, then \begin{equation} \frac{n_{out}}{n} \leq \frac{n_{rk}(1-h(\delta+\nu) - h(\delta))}{n} \leq \frac{1-2h(\delta)}{2} \end{equation} Suppose now that $n_{rk}/n \geq 1/2$. Simplify the problem by eliminating $\epsilon_{ec}$: note that
\begin{multline}
\epsilon_{ec}+\epsilon_{pa}(\nu) = \epsilon_{ec} + \frac{1}{2\sqrt{\epsilon_{ec}}} 2^{(-n_{rk}(1-h(\delta+\nu)-h(\delta))+n_{out})/2} \\
\geq \frac{3}{2^{4/3}} 2^{(-n_{rk}(1-h(\delta + \nu)-h(\delta))+n_{out})/3} 
\end{multline}
with equality if and only if \begin{equation} \epsilon_{ec} = \frac{1}{2^{4/3}} 2^{(-n_{rk}(1-h(\delta+\nu)-h(\delta))+n_{out})/3} \end{equation}

Use this and Lemma \ref{lemma:LowerBoundOnLimEtAlFunction} to deduce \begin{equation}  \frac{3}{2^{4/3}} 2^{(-n_{rk}(1-h(\delta + \nu)-h(\delta))+n_{out})/3} +  2 exp(-2 n_{pe} \nu^2) \leq \epsilon_{qkd}  \end{equation} From this, deduce further:
\begin{align}
-n_{rk}(1-h(\delta+\nu)-h(\delta)) + n_{out} &\leq 3 \log_2 \epsilon_{qkd} + 4 - 3\log_2(3) \\
-2n_{pe} \nu^2 &\leq \ln(\epsilon_{qkd}) - \ln(2)
\end{align}
Rewrite the first inequality as 
\begin{equation} \label{eq:UpperBoundOnOutputSize}
n_{out} \leq n(1-2h(\delta))- n_{pe}(1-2h(\delta)) -n_{rk}(h(\delta+\nu)-h(\delta)) \\
+3 \log_2 \epsilon_{qkd} + 4 - 3\log_2(3) 
\end{equation}
Now, apply the inequality $a+b \geq 3 a^{1/3} (b/2)^{2/3}$ to the second and third term:
\begin{multline}
n_{pe}(1-2h(\delta)) + n_{rk}(h(\delta+\nu)-h(\delta)) \\ \geq \frac{3}{2^{2/3}} n_{pe}^{1/3}(1-2h(\delta))^{1/3} n_{rk}^{2/3} (h(\delta+\nu)-h(\delta))^{2/3} \\ \geq \frac{3}{2^{2/3}} n_{pe}^{1/3}(1-2h(\delta))^{1/3} \left(n/2\right)^{2/3} (h(\delta+\nu)-h(\delta))^{2/3}
\end{multline}
Further, use the line through $(\delta,h(\delta))$ and $(1/2,1)$ to obtain \begin{equation} h(\delta+\nu)-h(\delta) \geq \nu \frac{1-h(\delta)}{1/2 - \delta} \end{equation} then combine this with $n_{pe} \nu^2 \geq 0.5 \ln(2/\epsilon_{qkd})$ to obtain \begin{equation} n_{pe}^{1/3} (h(\delta+\nu)-h(\delta))^{2/3} \geq \left(\frac{1-h(\delta)}{1/2 - \delta}\right)^{2/3} \left(\frac{1}{2} \ln \frac{2}{\epsilon_{qkd}}\right)^{1/3} \end{equation}
Thus, 
\begin{multline}
n_{pe}(1-2h(\delta)) + n_{rk}(h(\delta+\nu)-h(\delta)) \\ \geq \frac{3}{2^{5/3}} (1-2h(\delta))^{1/3} \left(\frac{1-h(\delta)}{1/2 - \delta}\right)^{2/3} \left( \ln \frac{2}{\epsilon_{qkd}}\right)^{1/3}  n^{2/3} 
\end{multline}
Combining with \eqref{eq:UpperBoundOnOutputSize} proves the Theorem. 
\end{proof}

\section{Conclusion and open problems}\label{sec:Conclusion}

The present paper has proposed and proved security of a QKD protocol that uses two-universal hashing instead of random sampling to perform parameter estimation. This protocol dramatically outperforms previous QKD protocols for small block sizes. This provides a new approach in QKD use-cases such as the Micius satellite example, where the difficulty of accumulating a large enough block size makes the BBM92 protocol impractical. The quantum phase of the two-universal hashing protocol is also impractical with current technology. However, it appears easier to make a moderate advance in ground stations and a moderate advance in the transmission of entangled photon pairs from space to earth, rather than to put the entire burden on only one of these approaches. 

The first group of open problems are related to the quantum phase of the two-universal hashing protocol. On the theoretical side, what are the fundamental trade-offs between the complexity of the quantum phase of a QKD protocol and its performance? Can good performance be achieved using a simpler quantum phase than the two-universal hashing protocol? Further, can QKD hardware be developed capable of more than just single qubit measurements? 

Second, the algorithm given in section \ref{sec:ComputingCertainFunctions} for computing the function $\functiong$ is not efficient. This leads to the following open problem: is there a probability distribution over CSS codes, such that the marginal distributions of the two parity check matrices satisfy a two-universal hashing condition with some good collision probability bound, and such that each of the two parity check matrices has additional structure that allows efficient computation of $\functiong$ during the protocol? There is a long history in information theory of approximating the performance of random codes with brute force decoding by more structured codes with efficient decoding, so there is reason to hope that the same can be done in the present case. 

Third, the arguments in the present paper are for the case where Alice and Bob can apply perfect quantum operations. It thus remains an open problem to generalize the present security proof to the case of imperfect devices.

\end{document}